\documentclass[conference]{IEEEtran}

\usepackage{graphicx}
\usepackage{amsmath}
\usepackage{amsfonts}
\usepackage{amsthm}
\usepackage{subfigure}
\usepackage{dsfont}
\usepackage{amssymb}
\usepackage{cases}

\usepackage[ruled,linesnumbered,vlined]{algorithm2e}

\newtheorem{theorem}{Theorem}

\newtheorem{corollary}{Corollary}
\newtheorem{lemma}{Lemma}

\begin{document}
\title{Online Vector Scheduling and Generalized Load Balancing}
\IEEEoverridecommandlockouts
\author{
\IEEEauthorblockN{Xiaojun Zhu\thanks{The work was done when the first author was visiting the College of William and Mary. This paper has been accepted to JPDC as a research note \cite{jpdc}. The current version contains more content than the published one due to page limitation of research notes of JPDC.}\IEEEauthorrefmark{1}\IEEEauthorrefmark{2}, Qun Li\IEEEauthorrefmark{2}, Weizhen Mao\IEEEauthorrefmark{2} and Guihai Chen\IEEEauthorrefmark{1}}
\IEEEauthorblockA{\IEEEauthorrefmark{1} State Key Laboratory for Novel Software Technology, Nanjing University, Nanjing, P. R. China}
\IEEEauthorblockA{\IEEEauthorrefmark{2} Department of Computer Science, the College of William and Mary, Williamsburg, VA, USA}
\IEEEauthorblockA{Email: gxjzhu@gmail.com, \{liqun,wm\}@cs.wm.edu, gchen@nju.edu.cn}
}

\maketitle
\pagenumbering{2}
\begin{abstract}
We give a polynomial time reduction from vector scheduling problem (VS)  to generalized load balancing  problem (GLB). This reduction gives the first non-trivial  online algorithm for VS where vectors come in an online fashion. The online algorithm is very simple in that each vector only needs to minimize the $L_{\ln(md)}$ norm of the resulting load when it comes, where $m$ is the number of partitions and $d$ is the dimension of vectors. It  has an approximation bound of $e\log(md)$, which is in $O(\ln(md))$, so it also improves the  $O(\ln^2d)$ bound of the existing polynomial time algorithm for VS. Additionally, the reduction  shows that GLB does not have constant approximation algorithms that run in polynomial time unless $P=NP$.
\end{abstract}

\section{Introduction}
Scheduling with costs is a very well studied problem in combinatorial optimization. The traditional paradigm assumes single-cost scenario: each job incurs a single cost to the machine that it is assigned to. The \emph{load} of a machine is the total cost incurred by the jobs it serves. The objective is to minimize the \emph{makespan}, the maximum machine load.
Vector scheduling and generalized load balancing extend the scenario in different directions.

Vector scheduling assumes that each job incurs a vector cost to the machine that it is assigned to.  The load of a machine is defined as the maximum
cost among all dimensions. The objective is to minimize the makespan. Vector scheduling is a multi-dimensional generalization of the traditional paradigm.
It finds application in multi-dimensional resource scheduling in parallel query optimization~\cite{mdpacking}. For example, a task may have requirements for CPU, memory and network at the same time, and this requirement is best described by a vector of CUP, memory and network, instead of an aggregate measure. In this scenario, the load of a server is also described by a vector.  To solve vector scheduling, there are three approximation solutions
 \cite{mdpacking}. Two of them are deterministic algorithms based on
derandomization of a randomized algorithm, with one providing $O(\ln ^2d)$ approximation\footnote{In this paper, $e$ denotes the natural number, $\ln(\cdot)$ denotes the natural logarithm,
and $\log(\cdot)$ denotes the logarithm base $2$.},
where $d$ is the dimension of vectors, and the other providing $O(\ln d)$ approximation with running time polynomial in $n^d$, where $n$ is the number of vectors. The third algorithm is a randomized algorithm, which assigns each vector to a uniformly and randomly chosen partition. It gives $O(\ln dm/\ln\ln dm)$ approximation with high probability, where $m$ is the number of partitions (servers).
For fixed $d$, there exists a polynomial time approximation scheme (PTAS)~\cite{mdpacking}. A PTAS has also been proposed for  a wide class of cost functions (rather than $\max$) \cite{general-vector-assignment}.

Generalized load balancing is recently introduced to model the effect of wireless interference~\cite{xuinfo}\cite{glb}.  Each job incurs costs to all machines, no matter which machine it is assigned to. The exact cost incurred by a job to a specific machine is dependent on which machine the job is assigned to. The load of a machine is the total cost incurred by all the jobs, instead of just the jobs it serves. This model is well suited for wireless transmission, since, in wireless network, a user may influence all APs in its transmission range due to the broadcast nature of wireless signal. To solve the generalized load balancing problem, the current solution is an online algorithm, adapted from the recent progress in online scheduling on traditional model~\cite{lpnormbetterbounds}. The solution, though provides good approximation, is rather simple: each job selects the machine to minimize the $L_\tau$ norm of the resulting loads at all machines where $\tau$ is a constant parameter to be optimized.
To avoid confusion, we keep the two terms \emph{job} and \emph{machine} unchanged for generalized load balancing, while refer to job and machine in the vector scheduling model as  \emph{vector} and \emph{partition} respectively.

We make two contributions. First, we present an approach to encode any vector scheduling instance by an instance of generalized load balancing problem (Section~\ref{sect:encode}). This encoding method  directly shows that generalized load balancing problem does not admit constant approximation algorithms unless $P=NP$.  Second, we design the first non-trivial online algorithm for vector scheduling based on the encoding method (Section~\ref{sect:online}). Directly applying the encoding method does not necessarily lead to a polynomial time algorithm, because it needs to compute the $L_{\ln l}$ norm function ($l$ is the number of machines), and it is unclear whether this norm can be computed in polynomial time. We eliminate this uncertainty by rounding $\ln l$ to the next integer, guaranteeing polynomial running time. In addition, we prove that the approximation loss due to rounding is  small.
We conclude this section by the following two definitions.

\subsection{Vector Scheduling}
 We are given positive integers $n,d,m$. There are a set $\mathcal V$ of $n$ rational and $d$-dimensional vectors $p_1,p_2,\ldots,p_n$ from $[0,\infty)^d$. Denote vector $p_i=(p_{i1},\ldots,p_{id})$. We need to partition the vectors in $\mathcal V$ into $m$ sets $A_1,\ldots,A_m$. The problem is to find a partition to minimize $\max_{1\le i\le m}\|\overline A_i\|_{\infty}$ where $\overline{A_i}=\sum_{j\in A_i}p_j$ is the sum of the vectors in $A_i$, and $\|\overline{A_i}\|_{\infty}$ is the infinity norm defined as the maximum element in the vector $\overline{A_i}$.
For the case $m\ge n$, there is a  trivial optimal solution that assigns vectors to distinct partitions. Therefore, we only consider the case $m<n$.

For ease of presentation, we give an equivalent integer program formulation. Let $x_{ij}$ be the indicator variable such that $x_{ij}=1$ if and only if vector $p_i$ is assigned to partition $A_j$. Then
\[
\|\overline A_j\|_{\infty}=\max_{1\le k\le d}\sum_i x_{ij}p_{ik}
\]

 The vector scheduling problem can be rewritten as

 \begin{equation*}
\begin{aligned}
\min_{\mathbf x}\max_{j,k}&\quad\sum_{i}x_{ij}p_{ik}\\
\text{subject to} \,&\\
& \sum_j x_{ij}=1,\quad\forall i\\
& x_{ij}\in \{0,1\},\quad \forall i,j
\end{aligned}
\qquad
\text{(VS)}
\end{equation*}

\subsection{Generalized Load Balancing}
This formulation first appears in \cite{glb}. We reformulate it with slightly different notations. There are a set $\mathcal M$ of independent machines, and a set $\mathcal J$ of jobs. If job $i$ is assigned to machine $j$, there is non-negative cost $c_{ijk}$ to machine $k$. The \emph{load} of a machine is defined as its total cost. The problem is to find an assignment (or schedule) to minimize the \emph{makespan}, the maximum load of all the machines.
This problem can be formally defined as follows.
\begin{equation*}
\begin{aligned}
\min_{\mathbf x}\max_k&\quad\sum_{i,j}x_{ij}c_{ijk}\\
\text{subject to} \,&\\
& \sum_j x_{ij}=1,\quad\forall i\in \mathcal J\\
& x_{ij}\in \{0,1\},\quad \forall i\in \mathcal J,j\in\mathcal M
\end{aligned}
\qquad
\text{(GLB)}
\end{equation*}
where $\mathbf x\in\{0,1\}^{|\mathcal J|\times|\mathcal M|}$ is the assignment matrix with elements $x_{ij}=1$ if and only if job $i$ is assigned to machine $j$. The two constraints require each job to be assigned to one machine.

\section{Encoding Vector Scheduling by Generalized Load Balancing}
\label{sect:encode}
We first create a  GLB instance for any VS instance, then prove their equivalence. At last, we discuss the hardness of GLB and  extend the VS model.

\subsection{Creating GLB Instances}
Comparing VS to GLB, we can find that they mainly differ in the subscripts of $\max$ and $\sum$. Our construction is inspired by this observation.

Given as input to VS the vector set $\mathcal V$ and $m$ partitions, we construct the GLB instance as follows. We set the jobs $\mathcal J=\mathcal V$. For each partition $A_j$ and its $k$-th dimension, we construct a machine, denoted by the pair $(j,k)$. Thus, the constructed machine set $\mathcal M$ is $\{(j,k)\mid j=1,2,\ldots,m\text{ and } k=1,2,\ldots,d\}$.  We refer to a machine as a pair of indices so that we can map the machine back to its corresponding partition and dimension easily. For a machine $t=(j,k)$ where $t\in \mathcal M$, we refer to the partition $j$ as $t_1$, and the dimension $k$ as $t_2$, i.e., $t=(t_1,t_2)$.
 We can see that there are totally $d$ machines ($t$ included) corresponding to the same partition as the machine $t$. We denote $[t]$ as the set of these machines,
  i.e., $[t]=\{(j,1),(j,2),\ldots,(j,d)\}$, where $j=t_1$. Among these $d$ machines, we select the first one $(j,1)$ as the \emph{anchor machine},
   denoted by $\overline t$, such that a vector chooses partition $A_j$ in VS if and only if the corresponding job chooses $\overline t$ in the new GLB problem.

The incurred cost $c_{ist}$ of job $i$ to machine $t$ if $i$ chooses machine $s$  is defined as
 \begin{numcases}{c_{ist}=}
p_{it_2}   & if  $s=\overline t$\label{eq1}\\
   \infty & if  $s\in [t]\wedge s\ne \overline t$\label{eq2}\\
 0 & if $s\notin [t]$\label{eq3}
 \end{numcases}
 where (\ref{eq1})(\ref{eq2}) are for the situation where $s$ and $t$ correspond to the same partition. They force a job to select only the anchor machines.  (\ref{eq3}) is for the situation where $s$ and $t$ correspond to different partitions. In this case, there is no load increase.

The resulting GLB instance is defined as VS-GLB:
\begin{equation*}
\begin{aligned}
\min_{\mathbf x'}\max_t&\quad\sum_{i,s}x'_{is}c_{ist}\\
\text{subject to} \,&\\
& \sum_s x'_{is}=1,\quad\forall i\in \mathcal J\\
& x'_{is}\in \{0,1\},\quad \forall i\in \mathcal J,s\in\mathcal M
\end{aligned}
\qquad
\text{(VS-GLB)}
\end{equation*}
To avoid the confusion with the general GLB problem, we intentionally use different notations $\mathbf x'$, $s$ and $t$. The notation $i$ is kept since it is in 1-1 correspondence with the vectors in VS.

As an example, consider the case when $d=1$. All vectors in VS have only one element, and there is only one machine in VS-GLB representing a partition in VS. The objective of VS becomes $\max_{j}\sum_i x_{ij}p_{i1}$. On the other hand, the objective of VS-GLB is $\max_t \sum_{i}x'_{it}c_{itt}=\max_t \sum_{i}x'_{it}p_{i1}$. Since any machine $t$ corresponds to a distinct partition $A_j$, simply changing subscripts shows that the two problems are equivalent. For the case when $d>1$, the proof is much involved, which we delay to Section~\ref{sect:equiv}.

\begin{theorem}
The construction of VS-GLB can be done in polynomial time.
\end{theorem}
\begin{proof}
An instance of VS needs $\Omega(nd)$ bits. The constructed VS-GLB instance has $n$ jobs, $md$ machines and $n(md)^2$ costs. Since $m<n$, all three terms are polynomials in $n$ and $d$. The theorem follows immediately.
\end{proof}

The following theorem shows that the constructed VS-GLB problem is equivalent to its corresponding VS problem. Let $T$ be a  positive constant.
\begin{theorem}
\label{thm:equ}
There is a feasible solution $\mathbf x$ to VS with objective value $T$ if and only if there is a feasible solution $\mathbf x'$ to VS-GLB with the same objective value $T$.
\end{theorem}

This theorem shows that VS and its corresponding VS-GLB  have the same optimal value. In addition, any $c$-approximation solution to VS-GLB, after transformation,  is also a $c$-approximation solution to VS, vice versa. We prove this theorem in Section~\ref{sect:equiv}.

It is worth mentioning that VS-GLB is a special instance of GLB. Since VS-GLB is converted from VS, VS is a special instance of GLB, which implies that VS should have approximation algorithms at least as good as GLB.
Unfortunately, on the contrary, the literature shows better approximation algorithm for GLB than that for VS.
Hence, it is worth applying algorithms of GLB to VS.

\subsection{Proof of Equivalence}
\label{sect:equiv}
We first study the properties of feasible solutions to VS-GLB in Lemma 1 and Lemma 2, and then prove Theorem~\ref{thm:equ}.

 \begin{lemma}
 \label{lem:first}
Given a feasible solution $\mathbf x'$ to VS-GLB yielding objective value $T$, for any $i\in \mathcal J$, we have
\begin{enumerate}
  \item $\forall s\ne \overline s, x'_{is}=0$;
  \item $\exists j$ such that for $s=(j,1), x'_{is}=1$.
\end{enumerate}
 \end{lemma}
 \begin{proof}
For 1), suppose $x'_{is}=1$ for some $s$ with $s\ne \overline s$. Then $x'_{is}c_{is\overline s}=\infty>T$, contradicting with $\max_{t} \sum_{i,s}x_{is}c_{ist}=T$.

For 2), since $\sum_s x'_{is}=1$, there exists some $s$ such that $x'_{is}=1$. Due to 1), we must have $s=\overline s$.
 \end{proof}

 \begin{lemma}
 \label{lem:second}
 Given a machine $t$, a job $i$, and a feasible solution $\mathbf x'$ to VS-GLB yielding objective value $T$, we have $\sum_{s}x'_{is}c_{ist}=x'_{i\overline{t}}p_{it_2}$.
 \end{lemma}
 \begin{proof}
Recall that $[t]=\{(t_1,1),(t_1,2),\ldots,(t_1,d)\}$. We have
 \begin{align}
 \sum_{s} x'_{is}c_{ist}&=\sum_{s\in [t]} x'_{is}c_{ist}+\sum_{s\notin [t]}x'_{is}c_{ist}\notag\\
 &=\sum_{s\in[t]} x'_{is}c_{ist}\label{eq:dkdk}\\
 &=x'_{i\overline t}c_{i\overline tt}\label{eq:sldkfds}\\
 &= x'_{i\overline t}p_{it_2}\label{eq:sldfjalsdkfj}
 \end{align}
 where (\ref{eq:dkdk}) is due to (\ref{eq3}), (\ref{eq:sldkfds}) is due to Lemma~\ref{lem:first}, and (\ref{eq:sldfjalsdkfj}) is due to (\ref{eq1}).
 \end{proof}

 With the two lemmas, we can now prove the equivalence.
 \begin{proof}[Proof of Theorem~\ref{thm:equ}]
 ``$\Longrightarrow$'' Given a feasible solution $\mathbf x$ to VS, construct a feasible solution $\mathbf x'$ to VS-GLB as follows. Set $x'_{i\overline s}=x_{is_1}$ and all others to be $0$. We first show that $\mathbf x'$ is a feasible solution to VS-GLB. Obviously, $\mathbf x'$ is an integer assignment. We will check that $\sum_{s}x'_{is}=1$. Observe that $x'_{is}=0$ if $s\ne \overline s$. We only need to consider $m$ machines $(1,1),(2,1),\ldots,(m,1)$. Since $\mathbf x$ is a feasible solution to VS, then for any $i\in \mathcal V$, there exists one and only one partition $A_j$ such that $x_{i,j}=1$. Our transformation sets $x'_{is}=1$ where $s=(j,1)$. So $\sum_{s}x'_{is}=1$.

 Second, we prove that the objective values of the two feasible solutions are equal.
 \begin{align}
 \max_t\sum_{i,s}x'_{is}c_{ist}&=\max_t\sum_{i\atop s\in [t]}x'_{is}c_{ist}\label{eq:first}\\
 &=\max_t \sum_{i}x'_{i\overline t}c_{i\overline tt}\label{eq:second}\\
 &=\max_t \sum_{i} x_{it_1}p_{it_2}\notag\\
 &=\max_{j,k} \sum_{i} x_{ij}p_{ik},\notag
 \end{align}
 where (\ref{eq:first}) is due to that $c_{ist}=0$ if $s\notin [t]$, and (\ref{eq:second}) is due to our assignment of $\mathbf x'$ that $x'_{is}=0$ if $s\ne \overline s$.

 ``$\Longleftarrow$'' Given $\mathbf x'$ for VS-GLB, construct $\mathbf x$ for VS as follows. Set $x_{ij}=x'_{i\overline s}$ where $s_1=j$. We show that $\mathbf x$ is a feasible solution to VS. Due to Lemma~\ref{lem:first}, for any $i$, there exists one $s$ such that $x'_{is}=1$ and $s=\overline s$. Therefore, there exists one $j$ such that $x_{ij}=1$. On the other hand, there cannot be two $j$s both with $x_{ij}=1$, otherwise $\mathbf x'$ is not a feasible solution to VS-GLB.

 For the objective value, we have
 \begin{align}
 \max_{j,k} \sum_{i} x_{ij}p_{ik}&=\max_{t=(j,k)}\sum_{i} x_{it_1}p_{it_2}\notag\\
 &=\max_t\sum_{i}x'_{i\overline t}p_{it_2}\notag\\
 &=\max_t\sum_{i,s}x'_{is}c_{ist},\label{eq:due22}
 \end{align}
where (\ref{eq:due22}) is due to Lemma~\ref{lem:second}. This completes our proof.

 \end{proof}

\subsection{Inapproximability for GLB}
It has been proved  that no polynomial time algorithm can give $c$-approximation solution to VS for any $c>1$ unless $NP=ZPP$ \cite{mdpacking}. Combining this result with Theorem~\ref{thm:equ}, we have the following theorem.
 \begin{theorem}
 \label{thm:apxhard}
For any constant $c>1$, there does not exist a polynomial time $c$-approximation algorithm  for GLB, unless $NP=ZPP$.
\end{theorem}
\begin{proof}
Since VS-GLB is a special instance of GLB, any $c$-approximation algorithm for GLB can be used to obtain $c$-approximation solution to VS-GLB. By Theorem~\ref{thm:equ}, any $c$-approximation solution to VS-GLB is also a $c$-approximation solution to the corresponding VS. Therefore, the approximation algorithm is also a $c$-approximation algorithm for VS, a contradiction.
\end{proof}

We can obtain a stronger result by relaxing the assumption $NP\ne ZPP$ to $P\ne NP$. (It is a relaxation because $P\subseteq ZPP\subseteq NP$.) This can be done by examining the inapproximability proof for VS \cite{mdpacking}. The inapproximability proof relies on the result that no polynomial time algorithm can approximate chromatic number to within $n^{1-\epsilon}$ for any $\epsilon>0$ unless $NP=ZPP$. Recently, it has been proved that no polynomial time algorithm can approximate chromatic number to within $n^{1-\epsilon}$ for any $\epsilon>0$ unless $P=NP$ \cite{chromaticnumber}. Thus, we can change the assumption $NP\ne ZPP$ to $P\ne NP$ safely.

\begin{theorem}
For any constant $c>1$, there does not exist a polynomial time $c$-approximation algorithm  for GLB, unless $P=NP$.
\end{theorem}

\subsection{Extending to generalized VS}
Our construction of VS-GLB and proof can be easily extended to a general version of vector scheduling. In the current VS definition, all machines (partitions) are identical so that any job (vector) incurs the same vector cost to all machines. The machines can be generalized to be heterogeneous so that each job incurs a different vector cost to different machines. Formally, job $i$ incurs vector cost $p^{(j)}_i$ to machine $j$ if $i$ is assigned to machine $j$. The formulation and transformations can be slightly changed as follows. In the integer program formulation of VS, change the objective to $\max_{j,k}\sum_i x_{ij}p^{(j)}_{ik}$. Change $p_{it_2}$ in equation (\ref{eq1}) to be $p_{it_2}^{(t_1)}$. For Lemma 2, change $x'_{i\overline t}p_{it_2}$ to $x'_{i\overline t}p_{it_2}^{(t_1)}$. It can be verified that the proof of Theorem~\ref{thm:equ} is still valid with minor modifications.
The online algorithm adopted later is also valid for this general version of vector scheduling. For simplicity, we mainly focus on the original VS model.

\section{Online Algorithm for VS}
\label{sect:online}
Based on Theorem~\ref{thm:equ}, we can solve VS by its corresponding VS-GLB. We review the approximation algorithm~\cite{xuinfo} for GLB, and then modify it to solve VS.

Given a GLB instance and a positive number $\tau$, the algorithm~\cite{xuinfo} considers jobs one by one (in an arbitrary order) and assigns the current job to a machine to  minimize the $L_\tau$ norm\footnote{$L_\tau$ norm of a vector $x=(x_1,x_2,\ldots,x_t)$ is defined as $(\sum_i x_i^\tau)^{1/\tau}$.} of the resulting load of all machines. Specifically, suppose jobs are numbered as $1,2,\ldots,n$, the same as the considered order. Suppose the load of machine $k$ after jobs $1,2,\ldots,i-1$ are assigned is $l^{i-1}_k$. Then job $i$ is assigned to the machine
\[
\arg\min_j \left(\sum_k{(l^{i-1}_k+c_{ijk})^\tau}\right)^{1/\tau}.
\]
The above optimization problem can be solved by trying each possible machine. During the optimization, the computation of the last step of $L_\tau$ norm, $(\cdot)^{1/\tau}$, can be omitted. In addition, because the algorithm does not require the order of jobs and each job is assigned once, it can be implemented in an online fashion. This algorithm was originally proposed  for the traditional load balancing problem \cite{lpnormbetterbounds}, and recently extended to the GLB problem  \cite{xuinfo}. The parameter $\tau$ controls the approximation ratio of the algorithm, as shown in the following lemma.

\begin{lemma}[\cite{lpnormbetterbounds,xuinfo}]
\label{lem:appratio}
Minimizing $L_\tau$ norm gives $\frac{\tau}{\ln(2)}l^{1/\tau}$ approximation ratio to solve GLB where $l$ is the number of machines.
\end{lemma}

Setting $\tau=\ln l$ yields the best approximation ratio $e\log l$. However, it is unclear whether the computation of $L_{\ln l}$ can be done in polynomial time. We consider this issue later.

\subsection{Adapting to VS}
To apply the above algorithm to VS, we can first solve VS-GLB and transform the solution to VS. This process can be simplified  by omitting the transformation between VS and VS-GLB.

 Recall that the algorithm is to assign vectors one by one. Consider a vector $p_i$ in VS.  To solve VS-GLB, this vector should choose a machine to minimize the $L_{\tau}$ norm of the resulting load. Due to the construction of VS-GLB, this vector can only choose from the \emph{anchor machines}, otherwise, the resulting $L_{\tau}$ norm would be infinite (definitely not the optimal choice).
Thus, this is equivalent to picking from the corresponding partitions in VS. After the assignment of any number of vectors that leads to partitions $A_1,A_2,\ldots,A_m$,
the $L_{\tau}$ norm of the  load of machines in VS-GLB
is, in fact, equal to
\[
f^{(\tau)}(A_1,\ldots,A_m)=\left(\|\overline A_1\|_{\tau}^\tau+\ldots+\|\overline A_m\|_{\tau}^\tau\right)^{1/\tau}
\]
where
\[\|\overline A_j\|_\tau^\tau=\sum_{k}\left(\sum_{i\in A_j}p_{ik}\right)^\tau.
\]
Suppose the assignment of vectors $p_1,p_2,\ldots,p_{i-1}$ leads to partitions $A_1,A_2,\ldots,A_m$. Let $f^{(\tau)}_{i,j}$ be $L_\tau$ norm of the resulting load if vector $p_i$ chooses partition $A_j$, i.e.,
\[
f^{(\tau)}_{i,j}=f^{(\tau)}(A_1,\ldots,A_j\cup\{p_i\},\ldots,A_m).
\]
Then, according to the algorithm, vector $p_i$ should be assigned to the partition
\[
\arg\min_j f^{(\tau)}_{i,j}.
\]

The procedure is described in Algorithm~\ref{online}. For each incoming vector, it only needs to execute Lines 5-9.

\begin{algorithm}
\KwIn{ $m$, the number of partitions; $d$, the dimension of each vector; $p_1,p_2,\ldots,p_n$, the $n$ vectors to be scheduled; $\tau$, the norm}
\KwOut{$A_1,\ldots,A_m$,  the $m$ partitions}
\Begin{
\For{$j$ from $1$ to $m$}{
$A_j\longleftarrow\emptyset$\;
}

\For{$i$ from $1$ to $n$}{
\eIf{$\exists A_j,A_j=\emptyset$}{
$A_j\longleftarrow A_j\bigcup\{p_i\}$\;
}
{
find $j$ to minimize $f^{(\tau)}_{i,j}$\;
$A_j\longleftarrow A_j\bigcup\{p_i\}$\;
}
}

}
\caption{Vector Scheduling} \label{online}
\end{algorithm}

Algorithm~\ref{online} with $\tau=\ln(md)$ is an $e\log(md)$ approximation algorithm to solve the corresponding VS-GLB.
Thus, we have the following result due to Theorem~\ref{thm:equ}.

\begin{lemma}
\label{lem:proof}
Algorithm~\ref{online} with $\tau=\ln(md)$ is an $e\log(md)$ approximation algorithm to solve VS.
\end{lemma}

However, it is unclear whether Algorithm~\ref{online} with $\tau=\ln(md)$ can terminate within polynomial time. The algorithm requires the computation of $x^{\ln (md)}$ for some $x$. First,
 the number $\ln(md)$ is irrational, thus cannot be represented by polynomial bits to achieve arbitrary resolution. Second, even though we can approximate it by a rational number with acceptable resolution, the number
 $x^{\tilde{\tau}}$ may still be irrational, where
   $\tilde{\tau}$ is the rational approximation to $\tau$. For example, when $\tilde{\tau}=1.5$, there are lots of values of $x$ such that $x^{1.5}$ are irrational. Though we can still approximate it by a rational number, it is complicated to theoretically analyze whether the approximation ratio still holds and how the running time increases with respect to rational number approximation accuracy.
 This  problem has not been addressed in literature.

Our solution is to round $\ln (md)$ to the next integer $\lceil \ln(md)\rceil$ and compute the $L_{\lceil \ln(md)\rceil}$ norm. This guarantees polynomial running time, but causes the loss of approximation ratio. We show in the following that the loss is very small.

\subsection{Guaranteeing Polynomial Running Time}
To deal with the irrational number issue, we round $\ln (md) $ to the next integer $\lceil \ln(md)\rceil$. In the following, we analyze the resulting approximation ratio.

\begin{theorem}
\label{thm:loss}
Let $l$ be the number of machines. Minimizing $L_{\lceil \ln l\rceil}$ norm gives $e\log(l)+\frac{e\log(e)}{\ln l+1}$ approximation ratio to solve GLB.
\end{theorem}
\begin{proof} This result is obtained from Lemma~\ref{lem:appratio} by performing calculus analysis.
Let $g(x)=\frac{x}{\ln(2)}l^{1/x}$.
Consider the derivative of $g$,
\[
g'(x)=\frac{l^{1/x}}{\ln 2}\left(1-\frac{\ln l}{x}\right).
\]
For $x\ge \ln l$, it holds that $g'(x)\ge 0$ so that the function $g(x)$ is monotonically increasing. Since $\ln(l)\le\lceil \ln (l)\rceil)\le \ln(l)+1$, we have
\[
g( \lceil \ln (l)\rceil)-g(\ln l)\le g(\ln l+1)-g(\ln l).
 \]
 In addition, consider the two points $(\ln l,g(\ln l))$ and $(\ln l +1,g(\ln l+1))$. Due to Langrange's mean value theorem in calculus, there exists $\xi\in[\ln l,\ln l+1]$ such that
\[
g(\ln l+1)-g(\ln l)=g'(\xi).
\]

Since $\xi\ge \ln l$, we have $l^{1/\xi}\le e$. Additionally, $\xi\le \ln l+1$, so $1-\frac{\ln l}{\xi}\le \frac{1}{\ln l+1}$. Therefore, $g'(\xi)\le \frac{e}{ \ln(2)}\left(1-\frac{\ln l}{\xi}\right)\le\frac{e\log(e)}{\ln l+1}$. We have
\begin{align*}
g( \lceil \ln l\rceil)-g(\ln l)&\le g(\ln l+1)-g(\ln l)\le \frac{e\log(e)}{\ln l+1}
\end{align*}

Note that $g(\ln l)=e\log(l)$. This completes our proof.
\end{proof}

This theorem holds for general GLB problem, such as the one considered in \cite{xuinfo} \cite{lpnormbetterbounds} and \cite{glb}. Of course, it holds for VS-GLB as well. To have an intuition on the loss, we plot the two approximation ratios  with respect to the number of machines in Figure~\ref{fig:loss}. We can see that the loss is small.
\begin{figure}[htpd]
\centering
\includegraphics[width=.35\textwidth]{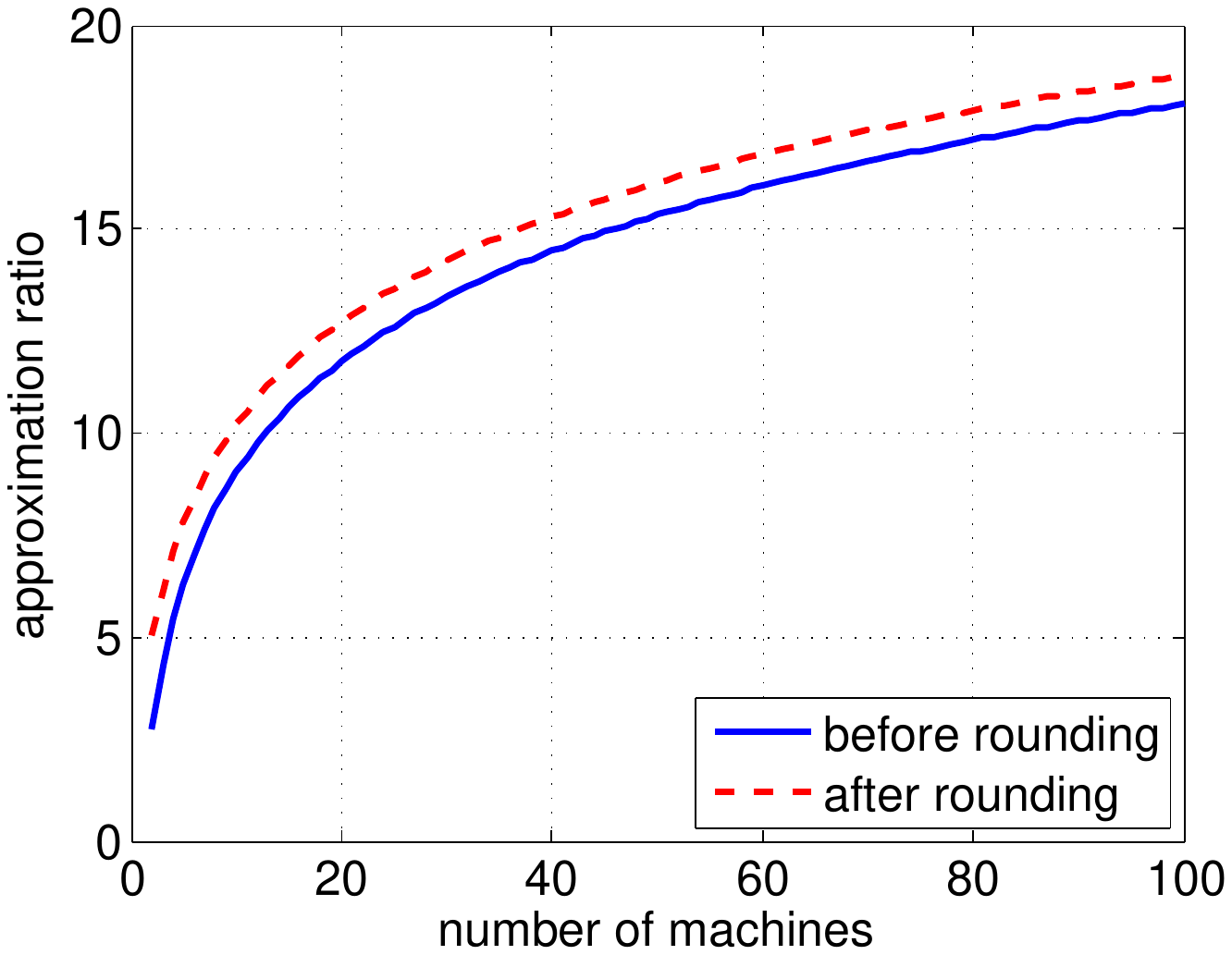}
\caption{Approximation ratio loss due to rounding. Before rounding, the approximation ratio is $y=e\log(x)$ and it becomes $y=e\log(x)+\frac{e\log(e)}{\ln(x)+1}$ after rounding.}
\label{fig:loss}
\end{figure}

We have the following corollary due to Theorem~\ref{thm:loss}.
\begin{corollary}
\label{co}
With $\tau=\lceil \ln (md)\rceil$, Algorithm~\ref{online} is an $e\log(md)+\frac{e\log(e)}{\ln(md)+1}$ approximation algorithm to VS, and it runs in polynomial time.
\end{corollary}

The polynomial running time can be shown by the following analysis.
We assume $\tau=\lceil \ln (md)\rceil$ if not specified.
The main time consuming step is to minimize $f^{(\tau)}_{i,j}$ over $j$ for given $i$. We can omit the computation of the outer $1/\tau$ power since function $x^{y}$ is monotonic for $x\ge 0$ and $y>0$. In computing $L_\tau$ norm, there is a basic operation, the integer power of a number, $a^\tau$, where $a$ is an element in any vector $\overline A_j$. The naive approach, which multiplies $a$ iteratively, involves $\tau-1$ multiplications. This can be improved by utilizing partial multiplication results. For example, computing $a^8$ as $((a^2)^2)^2$ only needs $3$ multiplications. Generally, computing $a^{\tau}$ requires $\lfloor \log(\tau) \rfloor+\nu(\tau)-1$ multiplications, where $\nu(\tau)$ is the number of $1$s in the binary representation of $\tau$ (Chapter 4.6.3 in~\cite{Knuth:1997:ACP:270146}). In the following, we put an upper bound $2\log\tau$ to the number of multiplications needed to compute $a^\tau$.

To compute $f^{(\tau)}_{i,j}$ for given $i$ and $j$, it needs $d+m-1$ additions (adding $p_i$ to $\overline A_j$, suppose  $\overline A_j$ is maintained in each iteration) and $2md\log(\tau)$ multiplications ($md$ numbers, each needs to compute its $\tau$ power). To find the optimal $j$ for given $i$, we need to compute $f^{(\tau)}_{i,j}$ for all $j$, and select the optimal one by comparison. This procedure needs $m(d+m-1)$ additions, $2d\log(\tau) m^2$ multiplications, and $m-1$ comparisons. In summary, it takes $O(d\log(\tau) m^2)$  time for one vector. For the overall algorithm, it takes $O(d\log(\tau) nm^2)$ time. The computations can be sped up by exploiting the problem structure. The complexity can be reduced to $O(d\log(\tau) mn)$, dropping one $m$ factor, as shown in the following.

\subsection{Computation Speedup}

Towards VS-GLB, we have the following lemma. Note that this lemma does not hold for the general GLB problem.
\begin{lemma}
\label{lem:speedup}
For any $j_1,j_2$, it holds that $f^{(\tau)}_{i,j_1}>f^{(\tau)}_{i,j_2}$ if and only if \[
\left\|\overline{A_{j_1}\cup\{p_i\}}\right\|_\tau^\tau-\left\|\overline{A_{j_1}}\right\|_\tau^\tau>\left\|\overline{A_{j_2}\cup\{p_i\}}\right\|_\tau^\tau-\left\|\overline{A_{j_2}}\right\|_\tau^\tau
\]
\end{lemma}
\begin{proof}
Adding $\|\overline A_1\|_{\tau}^\tau+\ldots+\|\overline A_m\|_{\tau}^\tau$ to both sides proves the lemma.
\end{proof}

Algorithm~\ref{speedup} shows the final design. For each partition $A_j$, the algorithm maintains two variables, the vector $\overline A_j$ ($\mu_j$ in the algorithm) and its norm $\|\overline A_j\|_\tau^\tau$ ($\delta_j$ in the algorithm). If there is no empty partition, then each incoming vector searches over all partitions to find the  $j$ to minimize $\left\|\overline{A_{j}\cup\{p_i\}}\right\|_\tau^\tau-\left\|\overline{A_{j}}\right\|_\tau^\tau$ (Lines 12-24). As Lemma~\ref{lem:speedup} shows, this is equivalent to minimize $f^{(\tau)}_{i,j}$.

For the running time, consider a new vector that cannot find an empty partition. There are $md$ additions (Lines 13,16), $2md\log(\tau)$ multiplications (Lines 14,17), $2(m-1)$ subtractions and $m-1$ comparisons (Line 18). The dominating factor is $md\log(\tau)$. This is for one vector. For all $n$ vectors, the running time is $O(mnd\log(\tau))$, compared to $O(m^2nd\log\tau)$ before speedup. Substituting $\tau=\lceil\ln(md)\rceil$ into the formula yields $O(nmd\ln\ln (md))$ running time, polynomial in the input length (note $m<n$). This analysis, together with Corollary~\ref{co} and Lemma~\ref{lem:speedup}, gives the following theorem.

\begin{theorem}
Algorithm~\ref{speedup}  is an $e\log(md)+\frac{e\log(e)}{\ln(md)+1}$ approximation algorithm to VS. It runs in $O(nmd\ln\ln (md))$ time.
\end{theorem}

It should be noted that we treat multiplications as basic operations in the above running time analysis. The running time will be different if we further consider the complexity of computing multiplications. Multiplying two $n$-bit integers takes time $O(n^{1.59})$ for a recursive algorithm (Chapter 5.5 in~\cite{Kleinberg:2005:AD:1051910}). Applying such analysis to Algorithm~\ref{speedup}, however, requires the consideration of the length of the binary representation of each numeric value in the vectors, which may be complicated.  Nevertheless, it is clear that multiplications run in polynomial time in the input length. Thus Algorithm~\ref{speedup} terminates certainly in polynomial time.

 \begin{algorithm}
\KwIn{ $m$, the number of partitions; $d$, the dimension of each vector; $p_1,p_2,\ldots,p_n$, the $n$ vectors to be scheduled}
\KwOut{$A_1,\ldots,A_m$,  the $m$ partitions}
\Begin{
\For{$j$ from $1$ to $m$}{
$A_j\longleftarrow\emptyset$\;
$\mathbf \mu_j\longleftarrow \mathbf 0$\tcp*{vector $\overline A_j$}
$\delta_j\longleftarrow 0$ \tcp*{scalar $\|\overline A_j\|_\tau^\tau$}
}
\For{$i$ from $1$ to $n$}{
\eIf{$\exists A_j,A_j=\emptyset$}{
$A_j\longleftarrow A_j\bigcup\{p_i\}$\;
$\mathbf \mu_j\longleftarrow p_i$\;
$\delta_j\longleftarrow \|p_i\|_\tau^\tau$\;
}
{
$j_{\min}\longleftarrow 1$\tcp*{partition index}
$\mathbf \mu_{\min}\longleftarrow \mathbf\mu_1+p_1$\;
$\delta_{\min}\longleftarrow  \|\mu_{\min}\|_\tau^\tau$\;
\For{$j$ from $2$ to $m$}{
$\tilde{\mu}\longleftarrow\mathbf\mu_j+p_i$\tcp*{vector addition}
$\tilde \delta\longleftarrow \|\tilde{\mu}\|_\tau^\tau$\;
\If{$\delta_{\min}-\delta_{j_{\min}}>\tilde \delta-\delta_j$}{
$j_{\min}\longleftarrow j$\;
$\mathbf \mu_{\min}\longleftarrow \tilde{\mathbf{\mu}}$\;
$\delta_{\min}\longleftarrow \tilde \delta$\;
}
}
$\mu_{j_{\min}}\longleftarrow \mu_{\min}$\;
$\delta_{j_{\min}}\longleftarrow \delta_{\min}$\;
$A_{j_{\min}}\longleftarrow A_{j_{\min}}\bigcup\{p_i\}$\;
}
}
}
\caption{Sped-up Vector Scheduling with $\tau=\lceil\ln (md) \rceil$} \label{speedup}
\end{algorithm}

\subsection{Simulations}
We implement three approaches for comparison: Algorithm~1 with $\tau=\ln (md)$, Algorithm~2 with $\tau=\lceil \ln (md) \rceil$, and a list scheduling algorithm mentioned in \cite{mdpacking}. The list scheduling algorithm is a $(d+1)$ approximation algorithm for vector scheduling. It ignores the multi-dimension property of vectors, and treats vectors as scalars equal to the summation of elements. We did not implement the $O(\ln^2 (d))$ approximation algorithm in  \cite{mdpacking} due to complicated implementation.

We consider two scenarios. In the first scenario, we study the approximation ratio of each algorithm. This requires the computation of the optimal solution, which is done by  enumerating all solutions and is time consuming, so we only consider small problem instances. Specifically, we consider problem instances with 3 machines (m=3), 10 jobs (n=10) and a dimension of 20 (d=20). For each job, its elements are drawn independently  from the uniform distribution in the range of $[0,1]$. Under such settings, the worst-case approximation ratios for Algorithm~1, Algorithm~2 and the list scheduling algorithm are 16.0566, 16.8264 and 21 respectively. We generate 100 problem instances and Figure~\ref{fig:apratio} shows the box plot of the approximation ratio of each algorithm. We can see that the empirical performance of every algorithm is much better than that suggested by the worst-case analysis, and Algorithms 1 and 2 outperform the list scheduling algorithm.

\begin{figure}[htpd]
\centering
\includegraphics[width=.35\textwidth]{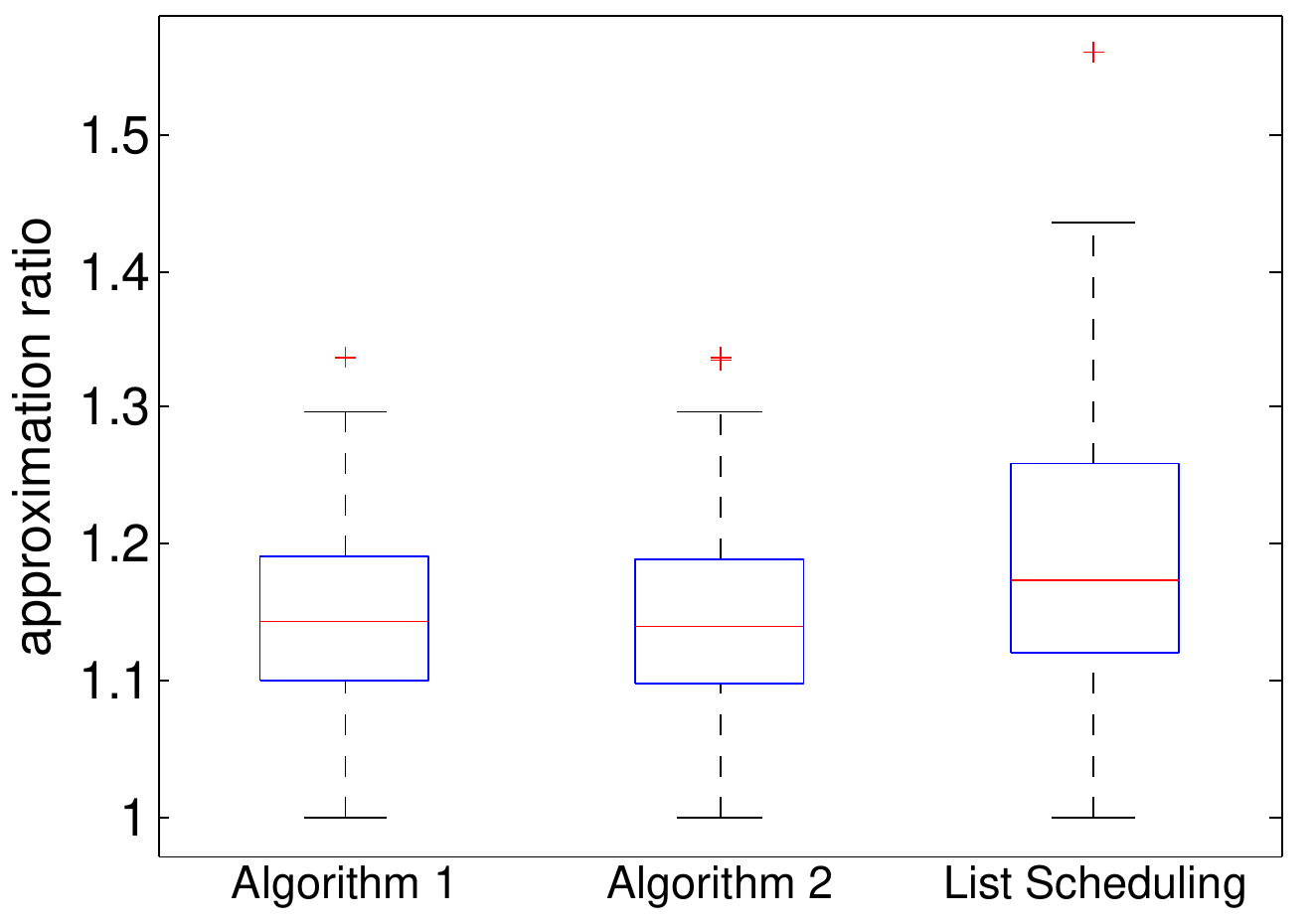}
\caption{Approximation ratio of three vector scheduling approaches. One hundred problem instances are generated to plot this figure.}
\label{fig:apratio}
\end{figure}

In the second scenario, we compare the three algorithms on larger problem instances. There are 10 machines  and 100 jobs.  The elements of a job are drawn from a uniform distribution as before. We vary the dimension $d$ from 10 to 40 with increments of 5. For each dimension, we generate 100 problem instances and compute the average makespan of the three approaches. Figure~\ref{fig:compare} shows that Algorithm~1 and Algorithm~2 perform similarly, and both of them greatly outperform the list scheduling algorithm. Note that with the increase of dimension, the makespan of all approaches increases. This is because the probability of an imbalanced dimension increases in this case.

\begin{figure}[htpd]
\centering
\includegraphics[width=.35\textwidth]{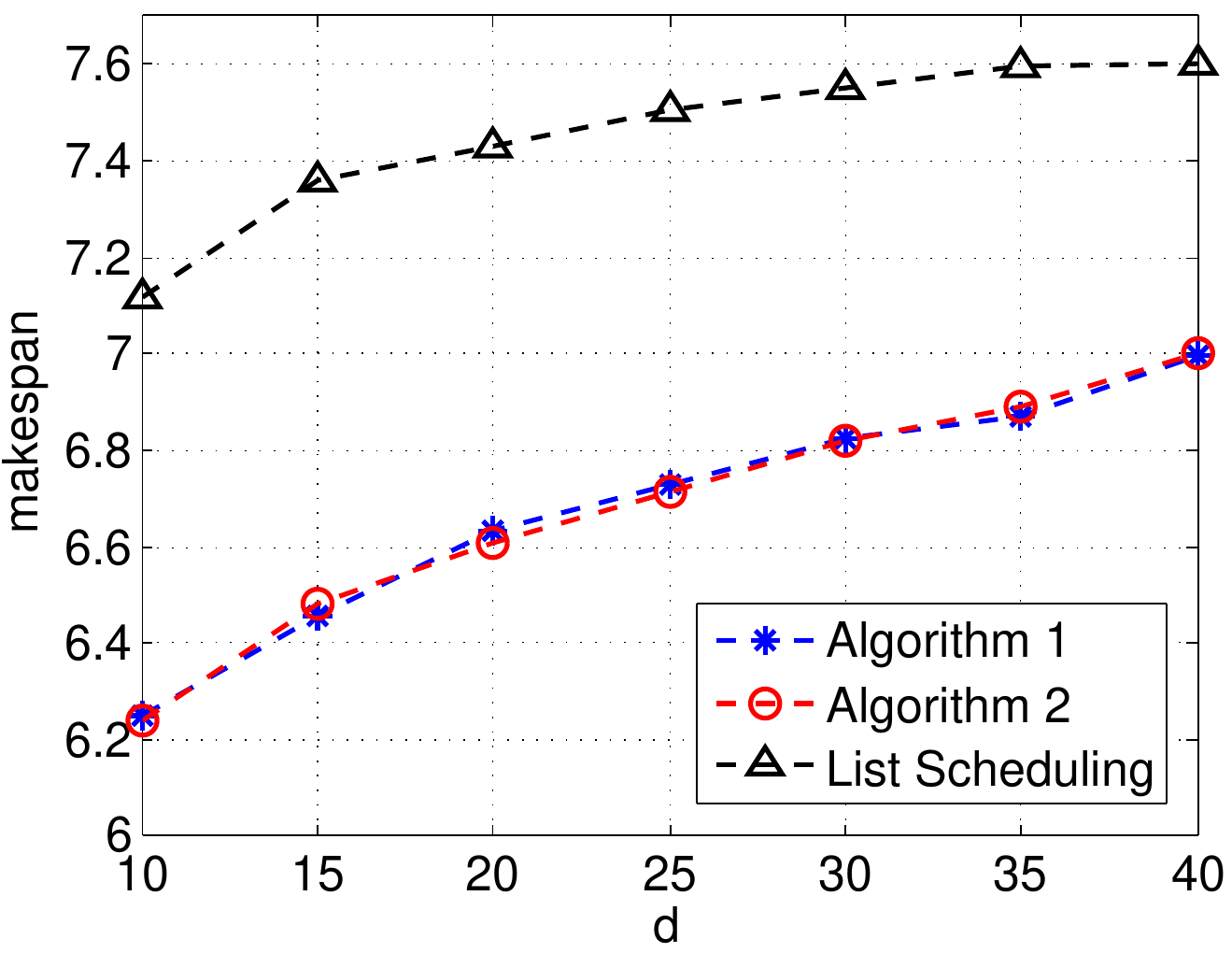}
\caption{Compare three vector scheduling approaches in terms of makespan. Each point in the figure is averaged over 100 problem instances.}
\label{fig:compare}
\end{figure}

\section{Conclusion}
In this work, we connect the vector scheduling problem with the generalized load balancing problem, and obtain new results by applying existing results to each other. Besides showing that generalized load balancing does not admit constant approximation algorithms unless $P=NP$, we give the first non-trivial online algorithm for vector scheduling. This online algorithm also provides better approximation bound to solve VS than existing offline polynomial time algorithm.

\bibliographystyle{IEEEtran}

\bibliography{IEEEabrv,zxj}

\end{document}